\documentclass[envcountsame]{llncs}

\usepackage[utf8x]{inputenc}
\usepackage[T1]{fontenc}
\usepackage{lmodern}
\usepackage[english]{babel}

\usepackage{amsmath}
\usepackage{amsxtra}
\usepackage{amstext}
\usepackage{amssymb}
\usepackage{latexsym}

\newcommand{\GF}[2][2]{{\mathbb F}_{#1^{#2}}}

\newcommand{\wt}{wt}

\newcommand{\V}[1]{{\mathbb F}^{#1}_2}

\begin{document}
\title{Improved upper bound on root number of linearized polynomials and its application to nonlinearity estimation of Boolean functions}
\author{Sihem Mesnager\inst{1}  \and Kwang Ho Kim\inst{2,3} \and Myong Song Jo\inst{4}} \institute{ LAGA, Department of Mathematics, University of Paris
VIII and Paris XIII,  CNRS and Telecom ParisTech, France
\email{smesnager@univ-paris8.fr} \and Institute of Mathematics,
State Academy of Sciences, Pyongyang, DPR Korea \and PGItech Corp.,
Pyongyang, DPR Korea \and  KumSong School, Pyongyang,
DPR Korea }

\maketitle

\begin{abstract}
To determine the dimension of null space of any given linearized
polynomial is one of vital problems in finite field theory, with
concern to design of modern symmetric cryptosystems. But, the known
general theory for this task is much far from giving the exact
dimension when applied to a specific linearized polynomial. The
first contribution of this paper is to give a better general method
to get more precise upper bound on the root number of any given
linearized polynomial. We anticipate this result would be applied as
a useful tool  in many research branches of finite field and
cryptography. Really we apply this result to get tighter estimations
of the lower bounds on the second order nonlinearities of general
cubic Boolean functions, which has been being an active research
problem during the past decade, with many examples showing great
improvements. Furthermore, this paper shows that by studying the
distribution of radicals of derivatives of a given Boolean functions
one can get a better lower bound of the second-order nonlinearity,
through an example of the monomial Boolean function $g_{\mu}=Tr(\mu
x^{2^{2r}+2^r+1})$ over any finite field $\GF{n}$.
\end{abstract}

\section{Introduction}
To determine the dimension of null space of linearized polynomials
is one of vital problems in finite field theory, with concern to
design of modern symmetric cryptosystems. But, the known general
theory for this task is much far from giving the exact dimension
when applied to a specific linearized polynomial. The first
contribution of this paper is to give a better general method to get
more precise upper bound on the root number of any given linearized
polynomial.

 As the second contribution we apply this result to get
tighter estimations of the lower bounds on the second order
nonlinearities of cubic Boolean functions, which has been being an
active research problem during the past decade as summarized below.

The $r-$th order nonlinearity of $n-$variable Boolean function $f$
is the minimum Hamming distance between $f$ and all $n-$variable
Boolean functions of degree at most $r$. Computing the $r$-th order
nonlinearity of a given function with algebraic degree strictly
greater than $r$ is a hard task for $r> 1$. Even the second-order
nonlinearity is unknown for all functions except for a few peculiar
ones and for functions in small numbers of variables. The best known
upper bound on the $r-$th nonlinearity for $r>1$ credits to Carlet
and Mesnager \cite{25}. Proving lower bounds on the $r$-th order
nonlinearity of functions is also a quite difficult task, even for
the second order \cite{5}.

In 2006, Carlet \cite{carlet07} and Carlet et al. \cite{carlet03} have presented  two lower bounds involving the algebraic immunity on the $r$th-order nonlinearity. None of them improves upon the other one in all situations. In 2007, the first author \cite{26} presented an improved lower bound on the
$r-$th-order nonlinearity profile of Boolean functions, given their
algebraic immunity.  Her results improve significantly upon the lower bound in  \cite{carlet03}
  for all orders  and upon the bound in \cite{carlet07} for low orders (which play the most important role for attacks). Note that relation between nonlinearity and algebraic immunity have been studied further in \cite{Lobanov2006,WangJohansson2010}.

In 2008, Carlet \cite{5} introduced a method to
determine the lower bound of the $r$-th order nonlinearity of a
function from the maximum value or the lower bounds of the
$(r-1)$-th order nonlinearity of its first derivatives, and obtained
the lower bounds on the second order nonlinearities of some
functions including Welch function and multiplicative inverse
function and so on. Carlet \cite{21} also lower bounded the
nonlinearity profile of the Dillon bent functions. In \cite{20},
Kolokotronis and Limniotis get a tighter lower bound on the
second-order nonlinearity of the cubic Boolean functions within the
Maiorana-McFarland class. In 2009, Sun and Wu \cite{10} have found
lower bounds of the second-order nonlinearities of three classes of
cubic bent Boolean functions, and Gangopadhyay, Sarkar and Telang
\cite{6} improved lower bounds on the second order nonlinearities of
the cubic monomial Boolean functions $Tr(\lambda x^{2^{2r}+2^r+1})$
over $\GF{n}$ with $n=6r$. Gode and Gangopadhyay \cite{7} lower
bound the second-order nonlinearities of the cubic monomial Boolean
functions. In 2010, Li, Hu, Gao \cite{8} extend these results from
monomial Boolean functions to Boolean functions with more trace
terms, and get better lower bound than those of Gode and
Gangopadhyay \cite{7} for monomial functions. In 2011,
 Singh \cite{17} lower bounded the second-order nonlinearity of
$Tr(\lambda x^{2^{2r}+2^r+1})$ over $\GF{n}$ with $n=3r$. Sun and Wu
\cite{16} obtained a better lower bound of second-order nonlinearity
of $Tr(\lambda x^{2^{2r}+2^r+1})$ over $\GF{n}$ with $n=4r$.
Gangopadhyay and Garg \cite{11} obtain a better lower bound of
second nonlinearity of $Tr(\lambda x^{2^{2r}+2^r+1})$ over $\GF{n}$
with $n=5r$. Garg and Gangopadhyay \cite{22} obtained a better lower
bound of second-order nonlinearity for a bent function via Niho
power function. In 2018, Carlet \cite{Carlet2018} has obtained an upper bound on the 
nonlinearity of monotone Boolean functions in even dimension and showed a
deep weakness of such functions.

In this paper, new results which significantly improve all these previous
estimations on lower bound of the second-order nonlinearity of
general cubic Boolean functions are achieved by applying the
improved upper-bound estimation of root number of linearized
polynomials, together with a set of examples.

Furthermore, this paper shows that one can get a better lower bound
of the second-order nonlinearity by studying the distribution of
radicals of derivatives of a given Boolean functions, by an example
of the Boolean function $g_{\mu}=Tr(\mu x^{2^{2r}+2^r+1})$ over any
finite field $\GF{n}$.

The paper is structured as follows. Section \ref{preliminaries}  sets  main notations and gives background on Boolean functions. In Section \ref{known_result}, we present the known lower-bounds on the second-order nonlinearity of Boolean functions. In Section \ref{Improved upper bound of root number on linearized polynomials}, new upper bound on the root number of linearized polynomials is given (Theorem \ref{corrootnum}). We also focus on the related Problem \ref{prob} and presents an algorithmic approach to this problem. In Section \ref{Application}, we apply the results of the previous sections to derive a better estimation on the second order nonlinearity of cubic Boolean functions  (Theorem \ref{maintheorem}). By examining examples, we show in Section \ref{Example} that our estimation  is more precise than the one given by Li, Hu and Gao \cite{8}. In Section \ref{better}, a deep analysis toward a better lower bound on the nonlinearity of cubic functions is presented as well as several open problems for future considerations.

\section{Preliminaries}\label{preliminaries}

Let $L$ be a Galois extension of a field $K$ and $\mathrm{Gal}
(L/K)$ be the Galois group of $L$ over $K$.  Let
$\sigma^0(x)=x,\sigma^j(x)=\sigma(\sigma^{j-1}(x))$ for $\sigma \in
\mathrm{Gal} (L/K)$ and $x\in L$. Then for a
 given polynomial $w(t)= \sum_{j=0}^l c_j t^j \in L[t]$, a homomorphism $w(\sigma)$
 is defined to act as $w(\sigma)x =
\sum_{j=0}^{l}c_j\sigma^j(x)$ on the element $x\in L$. The following
lemma characterizes the size of kernel space of the homomorphism
$w(\sigma)$.

\begin{lemma} (\cite{12,13}). Let
$L$ be a cyclic Galois extension of $K$ of degree $n$ and suppose
that $\sigma$ generates the Galois group of $L$ over $K$. Let $m$ be
an integer satisfying $1\leq m \leq n$ and $w(t)$ be a polynomial of
degree $m$ in $L[t]$. Let $R = \{x\in L | w(\sigma)x =0\}$. Then we
have $\dim_K R \leq m$.
\end{lemma}

Let $K=\GF{}$ and $L=\GF{n}$. Because given $\gcd(n,s)=1$,
$\sigma(x)=x^{2^s}$ is a generator of the Galois group of $L$ over
$K$, as a corollary we can get following.

\begin{lemma}\cite{2}\label{lemrootnum}
Let $g(x)=\sum_{i=0}^{\nu}r_ix^{2^{si}}(r_i\in \GF{n})$ be a
linearized polynomial over $\GF{n}$ with $\gcd(n,s)=1$. Then,
equation $g(x)=0$ has at most $2^\nu$ solutions in  $\GF{n}$.
\end{lemma}

A Boolean function $f$ is an $\V
{}$-valued function on the vectorspace $\V {n}$ over the prime field
$\V{}$ formed by all binary vectors of length $n$. We shall need a
representation of Boolean functions by univariate polynomials over
the Galois field $\GF{n}$ of order $2^n$. To this end, we identify
the field $\GF{n}$ with $\V {n}$ by choosing a basis of $\GF{n}$,
viewed as vector space over $\GF{}$. 
We denote the {\em absolute trace} over $\GF{}$ of an element $x \in
\GF {n}$ by $Tr_{1}^{n}(x) = \sum_{i=0}^{n-1} x^{2^i}$. The function
$Tr_{1}^n$ from $\GF{n}$ to its prime field $\GF{}$ is $\GF{}$-linear
and satisfies $(Tr_{1}^{n} (x))^2= Tr_{1}^{n}(x)=Tr_{1}^{n}(x^2)$ for every
$x\in \GF{n}$. The function $(x,y) \rightarrow Tr_{1}^{n}(xy)$ is an
inner product in $\GF{n}$. For any positive integer $k$, and $r$
dividing $k$, the trace function from $\GF{k}$ to $\GF{r}$, denoted
by $Tr_{r}^{k}$, is the mapping defined as:

\begin{displaymath}
  \forall x\in\GF{k} ,\quad
  Tr_{r}^{k}(x):=\sum_{i=0}^{\frac kr-1}
  x^{2^{ir}}=x+x^{2^r}+x^{2^{2r}}+\cdots+x^{2^{k-r}.}
\end{displaymath}

Recall that, for every integer $r$ dividing $k$, the trace function
$Tr_{r}^{k}$ satisfies the transitivity property.

.

Given an integer $e$, $0\leq e \leq 2^n-1$, having the binary
expansion: $e= \sum_{i=0}^{n-1} e_i2^i$, $e_i \in\{0,1\}$, the
2-weight of $e$, denoted by $w_2(e)$, is the Hamming weight of the
binary vector $(e_0, e_1,\cdots,e_{n-1})$.
Every non-zero Boolean function $f$ defined on $\GF{n}$ has a (unique)
trace expansion of the form:
\begin{equation}
  \forall x\in\GF n ,\quad
  f(x)=\sum_{j\in\Gamma_n}Tr_{1}^{{o(j)}}(a_jx^{j}) + \epsilon
  (1+x^{2^n-1}), \quad a_j \in
\GF{o(j)}
\end{equation}\label{traceform}
called its polynomial form, where $\Gamma_n$ is the set of integers
obtained by choosing one element in each cyclotomic class of $2$
modulo $2^n-1$, the most usual choice being the smallest element in
each cyclotomic class, called the coset leader of the class, and
$o(j)$ is the size of the cyclotomic coset containing $j$,
$\epsilon=\wt(f)$ modulo $2$. The algebraic degree of $f$, denoted by
$\deg (f)$, is equal to the maximum 2-weight of an exponent $j$ for
which $a_j\not= 0$ if $\epsilon=0$ and to $n$ if $\epsilon=1$.  Note
that $\epsilon=0$ when $\wt(f)$ is even, that is, when the algebraic
degree of $f$ is less than $n$. Note that when the integers modulo
$2^{n}-1$ are partitioned into cyclotomic classes of 2 modulo $2^n-1$,
all the elements in a cyclotomic class have the same 2-weight.

From now, we  shall denote $Tr$  the trace function from $\GF{n}$ to  $\GF{}$ defined by
$Tr(x)=x+x^2+x^{2^2}+\cdots+x^{2^{n-1}}$. 

A Boolean function on
$\GF{n}$ is a function can be expressed as $Tr(g[x])$, where $g[x]$ is any
polynomial in $\GF{n}[x]$. The Hamming weight of binary
representation of integer $\deg g[x]$ is the degree of
Boolean function $Tr(g[x])$ on $\GF{n}$. The (Hamming) distance between
Boolean functions $f_1$ and $f_2$ is defined by
$d(f_1,f_2)=\# \{\, x\in \GF{n} \, |\, f_1(x)\neq f_2(x) \,
\}$.

Let $f$ be any $n-$variable Boolean function on $\GF{n}$. The $r-$th
order nonlinearity of $f$, denoted by $nl_r(f)$, is the minimum
Hamming distance between $f$ and all $n-$variable Boolean functions
of degree at most  $r$, a nonnegative integer less than or equal to
$n$. The sequence of values $nl_r(f)$ for $r$ ranging from 1 to
$n-1$ is said to be the nonlinearity profile of  $f$. The first
order nonlinearity of $f$ is referred to as the nonlinearity of $f$
and denoted by $nl(f)$.

The Walsh transform of function $f$ at $u\in \GF{n}$ is defined by
\[
W_f(u)=\sum_{x\in \GF{n}}(-1)^{f(x)+Tr(ux)}, u\in \GF{n},
\]
and the Walsh spectrum of  $f$ as the set $\{\,W_f(u)\,|\,u\in
\GF{n}\,\}$. The nonlinearity and the Walsh
transform of $f$ are
 related as:
\begin{align}\label{Wal-nonli}
nl(f)=2^{n-1}-\frac{1}{2}\max_{u\in \GF{n}}|W_f(u)|.
\end{align}

The derivative of $f$ with respect to $b\in \GF{n}$ is the Boolean
function  $D_bf:x\mapsto f(x)+f(x+b)$. The kernel $\varepsilon_f$ of
quadratic Boolean function  $f$ is the $\GF{}-$linear subspace of
$\GF{n}$, defined by $\varepsilon_f=\{\,x\in \GF{n}\,|\,\forall y\in
\GF{n}, f(0)+f(x)+f(y)+f(x+y)=0\,\}$.

\begin{lemma}\cite{3}
Let  $f$ be any quadratic Boolean function. The kernel
$\varepsilon_f$ of $f$ is the subspace consisting of those  $b\in
\GF{n}$  such that the derivative $D_bf$ is constant.
\end{lemma}

\begin{lemma}\cite{3}\label{parity}
The dimension of the kernel  $\varepsilon_f$ of quadratic Boolean
function $f$  on $\GF{n}$ has the same parity as one of $n$.
\end{lemma}

\begin{lemma}\cite{3}\label{kerdim}
The Walsh Spectrum of quadratic Boolean function $f$  depends only
on the dimension $k$ of the kernel. The weight distribution of the
Walsh spectrum is
\begin{center}
\begin{tabular}{|c|c|}
\hline $W_f(u)$               &\text{Number of} $u\in
\GF{n}$\\\hline
 $0$                     &$2^n-2^{n-k}$\\ \hline
 $2^{\frac{n+k}{2}}$   &$2^{n-k-1}+(-1)^{f(0)}2^{\frac{n-k-2}{2}}$\\ \hline
 $-2^{\frac{n+k}{2}}$  &$2^{n-k-1}-(-1)^{f(0)}2^{\frac{n-k-2}{2}}$\\ \hline
\end{tabular}
\end{center}
\end{lemma}

\textbf{Note} Any quadratic Boolean form can be represented by
$Tr(\sum_{i=0}^{\lfloor\frac{n}{2}\rfloor}\delta_ix^{2^i+1}),
\delta_i\in \GF{n}$ \cite{9}.

Any cubic Boolean function over $\GF{n}$ can be written as
\begin{equation}\label{cubic}
f(x)=Tr(xQ(x))+Tr(xL(x))+a(x),
\end{equation}
where $Q$ is a quadratic polynomial, $L$ is a linearized polynomial
and $a$ is an affine Boolean function. Denote $\phi$ the polar form
associated to $Q$: $\phi(x,y)=Q(x+y)+Q(x)+Q(y)$.

Set $\widetilde{f}(x)=Tr(xQ(x))$ for every $x\in \GF{n}$. Note that
$nl_2(\widetilde{f})=nl_2(f)$. Now, for $a\in \GF{n}^*$,
\begin{align*}
D_a\widetilde{f}(x)&=Tr((x+a)Q(x+a)+xQ(x))\\
&=Tr(x\phi(a,x)+aQ(x))+Tr(xQ(a)+a\phi(a,x)+aQ(a)).
\end{align*}
Hence, $nl(D_a\widetilde{f})=nl(\psi_a)$, where for every $x\in
\GF{n}$
\[
\psi_a(x)=Tr(x\phi(a,x)+aQ(x)).
\]

By the relation \eqref{Wal-nonli} and Lemma \ref{kerdim}, the
nonlinearity of a nonzero quadratic form can be expressed in terms
of its radical:
\[
nl(\psi_a)=2^{n-1}-2^{\frac{n+r_a}{2}-1}
\]
where $r_a$ is the dimension of the vector space
$\varepsilon_{f,a}:=\{x\in \GF{n}|\forall y\in\GF{n}, B_a(x,y)=0\}$
over $\GF{}$, i.e. the radical of  $\psi_a$, where $B_a$ is the
polar form of $\psi_a$: $B_a=a\phi(x,y)+x\phi(a,y)+y\phi(a,x)$. Note
always $a\in\varepsilon_{f,a}$ and therefore
\begin{equation}
r_a\geq 1, \text{   for every   } a\in \GF{n}^*.
\end{equation}
The reader can consult \cite{Cbook1} for more background on Boolean functions.

\section{Known results on the lower bounds on the second-order nonlinearity of Boolean functions }\label{known_result}
 Let us now recall the following lower
bound on the second-order nonlinearity of Boolean functions. Let $f$
be any Boolean function on $\GF{n}$  and $r$  a positive integer
smaller than  $n$.
\begin{theorem}\cite{5}
\begin{equation}
nl_r(f)\geq \frac{1}{2}\max_{a\in \GF{n}}nl_{r-1}(D_af).
\end{equation}
\end{theorem}

\begin{theorem}\cite{5}\label{r-nonli-r-1}
\begin{equation}
nl_r(f)\geq 2^{n-1}-\frac{1}{2}\sqrt{2^{2n}-2\sum_{a\in
\GF{n}}nl_{r-1}(D_af)}.
\end{equation}
\end{theorem}

If we apply these lower bounds to a cubic function of the form
\eqref{cubic}, we get
\[
nl_2(f)\geq \max\left(
\frac{1}{2}\max_{a\in\GF{n}^*}(2^{n-1}-2^{\frac{n+r_a}{2}-1}),
2^{n-1}-\frac{1}{2}\sqrt{2^{2n}-2\sum_{a\in\GF{n}}(2^{n-1}-2^{\frac{n+r_a}{2}-1})}\right),
\]
or,
\begin{equation}\label{Mesnager}
nl_2(f)\geq
\max\left(2^{n-2}-\frac{1}{4}\min_{a\in\GF{n}^*}2^{\frac{n+r_a}{2}},
2^{n-1}-\frac{1}{2}\sqrt{2^n+\sum_{a\in\GF{n}^*}2^{\frac{n+r_a}{2}}}\right).
\end{equation}

From \eqref{Mesnager}, immediately it follows:
\begin{corollary}\cite{5}\label{cor6}
For any cubic Boolean function $f$ no possessing affine
derivatives,
\begin{equation}
nl_2(f)\geq 2^{n-1}-2^{n-\frac{3}{2}}
\end{equation}
\end{corollary}

Gode and Gangopadhyay \cite{7} have improved on this for monomial
Boolean functions:
\begin{theorem}\cite{7}\label{Gode-Gango-F}
Let $f_\mu(x)=Tr(\mu x^{2^i+2^j+1})$, where $\mu \in \GF{n}$, and
$i,j$ are integers such that $n>i>j>0$.

For  $n>2i$, if $n$  is an even, then
\begin{equation}
nl_2(f_\mu)\geq
2^{n-1}-\frac{1}{2}\sqrt{2^n+(2^n-1)2^{\frac{n+2i}{2}}},
\end{equation}
and if  $n$ is an odd, then
\begin{equation}
nl_2(f_\mu)\geq
2^{n-1}-\frac{1}{2}\sqrt{2^n+(2^n-1)2^{\frac{n+2i-1}{2}}}.
\end{equation}
\end{theorem}

\begin{theorem}\cite{7}\label{Gode-Gango-G}
Let $g_\mu(x)=Tr(\mu x^{2^{2r}+2^r+1})$, where $\mu \in \GF{n}$ and
$\gcd (n,r)=1$.

For $n>3$, if $n$  is an even, then
\begin{equation}
nl_2(g_\mu)\geq
2^{n-1}-\frac{1}{2}\sqrt{2^n+(2^n-1)2^{\frac{n+4}{2}}},
\end{equation}
and if  $n$ is an odd, then
\begin{equation}
nl_2(g_\mu)\geq
2^{n-1}-\frac{1}{2}\sqrt{2^n+(2^n-1)2^{\frac{n+3}{2}}}.
\end{equation}
\end{theorem}

Li, Hu and Gao \cite{8} have improved on Corollary \ref{cor6} for
general cubic Boolean functions, while for cubic monomial Boolean
functions the improved estimation are better than ones given in
Theorem \ref{Gode-Gango-F}:

\begin{theorem}\cite{8}\label{Li-Hu-Gao-F}
Let  $F_\mu=Tr(\sum_{l=1}^{m}\mu_l x^{d_l})$, where $\mu_l\in
\GF{n}$ and $d_l=2^{i_l+j_l+1}$, $n>i_l>j_l>0$. Let us suppose that
any derivative of $F_\mu$ be a quadratic function. Let
$h_u(x)=Tr(\sum_{i=1}^{n-1}c_{i,u}x^{2^i+1}), c_{i,u}\in \GF{n}$, be
the quadratic part of the derivative of $F_\mu$ at $u\in \GF{n}$.

Let $s=\min \{\,i\,|\,\exists u, c_{i,u}\neq 0, 1\leq i \leq
n-1\,\}$, $t=\max \{\,i\,|\,\exists u\in \GF{n}, c_{i,u}\neq 0,
1\leq i \leq n-1\,\}$ and $t_1=\max \{\,i\,|\,\exists u\in \GF{n},
c_{i,u}\neq 0, i\neq t\,\}$ if $s\neq t$ or $n\neq 2t$.

\textcircled{1} If $n<s+t$,
\begin{equation}
nl_2(F_\mu)\geq 2^{n-1}-\frac{1}{2}\sqrt{2^n+(2^n-1)2^t},
\end{equation}

\textcircled{2}  If  $2t>n\geq s+t$,
\begin{equation}
nl_2(F_\mu)\geq 2^{n-1}-\frac{1}{2}\sqrt{2^n+(2^n-1)2^{n-s}},
\end{equation}

\textcircled{3} If $n=2t$  and  $s\neq t$, let
$p=\min\{n-2s,2t_1\}$,
\begin{equation}
nl_2(F_\mu)\geq
2^{n-1}-\frac{1}{2}\sqrt{2^n+(2^n-1)2^{\frac{n+p}{2}}},
\end{equation}

\textcircled{4}    If   $n>2t$ is an even, let  $p=\min\{n-2s,2t\}$,
\begin{equation}
nl_2(F_\mu)\geq
2^{n-1}-\frac{1}{2}\sqrt{2^n+(2^n-1)2^{\frac{n+p}{2}}},
\end{equation}

If  $n>2t$ is an odd, let  $q=\min\{n-2s,2t-1\}$,
\begin{equation}
nl_2(F_\mu)\geq
2^{n-1}-\frac{1}{2}\sqrt{2^n+(2^n-1)2^{\frac{n+q}{2}}}.
\end{equation}
\end{theorem}

 Li, Hu and Gao also generalized the
Gode-Gangopadhyay estimation for cubic monomial Boolean functions
$g_\mu$ (Theorem \ref{Gode-Gango-G}) to cubic Boolean functions
$G_\mu=Tr(\sum_{l=1}^{m}\mu_l x^{d_l})$, where $\mu_l\in \GF{n}$ and
$d_l=2^{i_lr+j_lr+1}$, $i_l>j_l>0$, $\gcd (n,r)=1, r \neq 1$.

\begin{theorem}\cite{8}\label{Li-Hu-Gao-G}
Let  $t=\max\{\,i_l\,|\,1\leq l \leq m\,\}$. Let us suppose that any
derivative of $G_\mu$ be quadratic function. For  $n\geq 2t$, if
$n$ is an even, then
\begin{equation}
nl_2(G_\mu)\geq
2^{n-1}-\frac{1}{2}\sqrt{2^n+(2^n-1)2^{\frac{n+2t}{2}}}.
\end{equation}
And if $n$  is an odd, then
\begin{equation}
nl_2(G_\mu)\geq
2^{n-1}-\frac{1}{2}\sqrt{2^n+(2^n-1)2^{\frac{n+2t-1}{2}}}.
\end{equation}
\end{theorem}

Note that Theorem \ref{Li-Hu-Gao-G} restricted to $g_\mu$ coincides
with Theorem \ref{Gode-Gango-G} and (a generalization of) this is
reformulated as Corollary 5 in \cite{8}.
\section{On the root number of linearized polynomials}\label{Improved upper bound of root number on linearized polynomials}
In this section, we present an improvement of the upper bound on the root number of linearized polynomials as well as an algorithmic solution of Problem \ref{prob}.

\subsection{Improved upper bound on the root number of linearized polynomials}

To begin with, recall some simple facts which are found in
elementary number theory.
\begin{definition}
Let $p$ be a prime. The $p-$adic norm (or, also called $p-$adic
valuation) of a rational number $d=p^r\frac{B}{A}$, where $A, B \in
\mathbb{Z}$ and $\gcd(A,p)=\gcd(B,p)=1$, is denoted by $\|d\|_p$ and
defined by $\|d\|_p=p^{-r}$.
\end{definition}

\begin{definition}
We define a function $gg: \mathbb{Z}^*\times
\mathbb{Z}^*\longrightarrow \mathbb{Z}^*$ by
$gg(A,B)=\frac{1}{\prod_{p|B:\text{ prime}}\|A\|_p}$.
\end{definition}

\begin{proposition}\label{prop}
For any two nonzero integers $A$ and $B$, followings are facts.
\begin{enumerate}
\item $\gcd(A,B)|gg(a,B)$. In particular, $\gcd(A,B)\leq gg(a,B)$.
\item $gg(A,B)$
and $gg(B,A)$ have the same prime factors, and $\gcd(gg(A,B),
gg(B,A))=\gcd(A, B)$.
\item the value $\frac{A}{gg(A,B)}$ is
an integer and it holds
\[
\gcd(\frac{A}{gg(A,B)}, B)=1.
\]
In fact, $\frac{A}{gg(A,B)}$ is the greatest divisor of $A$ that is
coprime to $B$.
\item\label{propit4} If $A$ divides $A'$, then $\frac{A}{gg(A,B)}$ divides $\frac{A'}{gg(A',B)}$.
\end{enumerate}
\end{proposition}

Then we are going to deduce an improved upper bound estimation on
numbers of roots of linearized polynomials.

\begin{lemma}\label{lemrooteq}
Let  $r_1<r_2$ be integers. Any linearized polynomial
$L(x)=\sum_{i=r_1}^{r_2}\alpha_ix^{2^i}(\alpha_i\in \GF{n})$ over
$\GF{n}$ has the same number of roots in $\GF{n}$ as
$L'(x)=\sum_{i=r_1}^{r_2}\alpha_i^{2^k}x^{2^{i+k+k_in}}$ has in
$\GF{n}$, where $k,k_i(i\in\overline{\{r_1,r_2\}})$ are arbitrarily
given integers.
\end{lemma}
\begin{proof}
 $x\in \GF{n}$ is a root of  $L(x)$ $\Longleftrightarrow$
 $L(x)=0$
$\Longleftrightarrow$ $L(x)^{2^k}=0$ $\Longleftrightarrow$\\
$\sum_{i=r_1}^{r_2}\alpha_i^{2^k}x^{2^{i+k}}=0$
$\Longleftrightarrow$
$\sum_{i=r_1}^{r_2}\alpha_i^{2^k}x^{2^{i+k+k_in}}=0$ \\(Regarding to
$x^{2^{k_in}}=x$ which follows from  $x\in \GF{n}$)
\\$\Longleftrightarrow$ $x\in \GF{n}$  is a root of $L'(x)$.
\end{proof}

\begin{theorem}\label{corrootnum}
Let  $r_1<r_2$ be integers and
$L(x)=\sum_{i=r_1}^{r_2}\alpha_ix^{2^i}(\alpha_i\in \GF{n})$ be a
linearized polynomial over $\GF{n}$. Let us introduce following
notations: $\Delta=\{\,i\,|\,\alpha_i\neq 0, r_1\leq i \leq
r_2\,\}=\{i_0,i_1,\cdots,i_{t-1}\}$ and
$U=\{\,K=(k,k_0,k_1,\cdots,k_{t-1})\in \mathbb{Z}^{t+1}\,|\,\forall
j\in \overline{\{0,t-1\}}, i_j+k+k_jn\geq0\,\}$. For $K\in U$, let
us define following quantities sequentially:
$T_K=\gcd(\{i_j+k+k_jn|j\in\overline{\{0,t-1\}}\})$, $S_K=T_K/
gg(T_K,n)$,
$V_K=\max_{j\in\overline{\{0,t-1\}}}\{\frac{i_j+k+k_jn}{S_K}\}$ and
$V=\min_{K\in U}V_K$.

Then $L(x)$ has at most $2^{V}$ solutions in $\GF{n}$.
\end{theorem}
\begin{proof} By Lemma \ref{lemrooteq}, we know that the number of  $\GF{n}$-roots of $L(x)$
equals to the number of  $\GF{n}$-roots of
$L'(x)=\sum_{i=r_1}^{r_2}\alpha_i^{2^k}x^{2^{i+k+k_in}}$ for any
$K=(k,k_0,\cdots,k_{t-1})\in U$. $L'(x)=\sum_{i\in
\Delta}\alpha_i^{2^k}x^{2^{S_K\cdot\frac{i+k+k_in}{S_K}}}=\sum_{l=0}^{V_K}\beta_lx^{2^{S_K\cdot
l}}$, where $\beta_l=\sum\alpha_i^{2^k}$ and the sum is over all
$i\in \Delta$ such that $l=\frac{i+k+k_in}{S_K}$. (If there no
exists such $i\in \Delta$, then we think $\beta_l=0$.) Since
$\gcd(S_K,n)=1$ by Proposition \ref{prop}, Lemma \ref{lemrootnum}
says that the number of $L'(x)$'s roots belonging to $\GF{n}$ is not
greater than $2^{V_K}$, so that the number of $L(x)$'s roots
belonging to $\GF{n}$ is not greater than $2^{V_K}$, from which the
theorem are validated.
\end{proof}

\subsection{Search for the Minimum $V$}\label{secmini}
In this subsection, we consider following problem.

\begin{problem}\label{prob} Given an integer $n$ and an integer set $\Delta=\{i_0,i_1,\cdots,i_{t-1}\}$,
where $n>i_0>i_1>\cdots>i_{t-1}$ be assumed, and let
$U=\{\,K=(k,k_0,k_1,\cdots,k_{t-1})\in \mathbb{Z}^{t+1}\,|\,\forall
j\in \overline{\{0,t-1\}}, i_j+k+k_jn\geq0\,\}$. For
$K=(k,k_0,k_1,\cdots,k_{t-1})\in U$, let us define
$T_K=\gcd(\{i_j+k+k_jn|j\in\overline{\{0,t-1\}}\})$, $S_K=T_K/
gg(T_K,n)$,
$V_K=\max_{j\in\overline{\{0,t-1\}}}\{\frac{i_j+k+k_jn}{S_K}\}$ and
$V=\min_{K\in U}V_K$. Find a $K$ such that $V_K=V$.
\end{problem}

Seemingly,  it looks like one has to scan the infinite space $U$ to
solve this problem. But, below we show that there exists a
polynomial-time algorithm to solve this problem.

To begin with, we have following useful fact:
\begin{proposition}\label{prop0}
For every $K=(k,k_0,\cdots,k_{t-1})$ attaining the minimum $V=V_K$
to be found in Problem \ref{prob},
\[
\min_{j\in\overline{\{0,t-1\}}}\{i_j+k+k_jn\}=0.
\]
\end{proposition}
\begin{proof}
Let us assume the opposition:
$\min_{j\in\overline{\{0,t-1\}}}\{i_j+k+k_jn\}\neq 0$ (i.e. $> 0$).
We can assume wlog that
$\min_{j\in\overline{\{0,t-1\}}}\{i_j+k+k_jn\}=i_0+k+k_0n$. Let us
set $k'=-i_0-k_0n$ and $K'=(k',k_0,\cdots,k_{t-1})$. Then, because
$i_j+k'+k_jn=i_j-i_0-k_0n+k_jn=(i_j+k+k_jn)-(i_0+k+k_0n)$ for every
$j\in\overline{\{0,t-1\}}$, it holds $T_K|T_{K'}$ and so $S_K\leq
S_{K'}$ by the item \ref{propit4} of Proposition \ref{prop}. Also,
since $i_j+k'+k_jn=(i_j+k+k_jn)-(i_0+k+k_0n)<(i_j+k+k_jn)$ for every
$j$, we get
\[
V_{K'}=\max_{j\in\overline{\{0,t-1\}}}\{\frac{i_j+k'+k_jn}{S_{K'}}\}<\max_{j\in\overline{\{0,t-1\}}}\{\frac{i_j+k+k_jn}{S_{K}}\}=V_K,
\]
which is a contradiction to the assumption that $K$ attains the
minimum $V=V_K$.
\end{proof}

On the other hand, since $K'=(k\mod n,
k_0+\lfloor\frac{k}{n}\rfloor, \cdots,
k_{t-1}+\lfloor\frac{k}{n}\rfloor)$ gives the same $T_K,S_K,V_K$ as
$K=(k,k_0,\cdots,k_{t-1})$ gives, i.e.
$T_{K'}=T_K,S_{K'}=S_K,V_{K'}=V_K$, though there are infinite number
of $K$'s such that $V_K=V$, we can restrict the range of $k$ into
the sub-opened interval $[0,n)$. Further specifically, by making use
of the assumption $n>i_0>i_1>\cdots>i_{t-1}$ and Proposition
\ref{prop0}, we can restrict the range of $k$ into the set
$k_S=\{(n-i_j)\mod n\}_{j\in\overline{\{0,t-1\}}}$.

Denote $V_0=(i_0-i_{t-1})\mod n$. Letting
$K_0=(-i_{t-1},-\lfloor\frac{i_0-i_{t-1}}{n}\rfloor,\cdots,-\lfloor\frac{i_{t-2}-i_{t-1}}{n}\rfloor,0)$,
we have $K_0\in U$ and $V_{K_0}\leq V_0$, and therefore it follows
\[
V\leq V_0<n.
\]
Let us introduce denotations $L_{j}=\frac{i_j+k+k_jn}{S_K}, 0\leq
j\leq t-1$ and $a=S_K^{-1}\mod n$ (This value exists because
$\gcd(S_K,n)=1$). It is true $L_{j}\mod n = a(i_j+k)\mod n$. Also,
we know that if $K$ is a solution to Problem \ref{prob}, then $0\leq
L_{j}\leq V_K=V<n, 0\leq j\leq t-1$, and therefore identically
\[
L_{j}=a(i_j+k)\mod n, 0\leq j\leq t-1.
\]

With all these information, we are reduced to explore all possible
$\phi (n)$ $a$'s, i.e. such as $\gcd(a,n)=1$, where $\phi$ is Euler
Phi-function.

\textbf{Algorithm searching for a $K$ attaining the minimum $V$}\\
1. $V\leftarrow (i_0-i_{t-1})\mod n$;\\
2. \verb"For" $index=0$ up to $t-1$;\\
3. $k\leftarrow(n-i_{index})\mod n$;\\
4. \qquad \verb"For" $a=1$ up to $n-1$; \\
5. \qquad\qquad Compute $d=\gcd(a,n)$;\\
6. \qquad\qquad \verb"If" $d=1$ \verb"Then";\\
7. \qquad\qquad\qquad \verb"For" $j=1$ up to $t-1$;\\
8. \qquad\qquad\qquad\qquad $L_j\leftarrow (a\times (k+i_j))\mod n$;\\
9. \qquad\qquad\qquad \verb"End For";\\
10. \qquad\qquad\qquad \verb"If" $V > \max_j L_j$ \verb"Then";\\
11. \qquad\qquad\qquad\qquad $V\leftarrow \max_j L_j$;\\
12. \qquad\qquad\qquad\qquad $a'\leftarrow a^{-1}\mod n$;\\
13. \qquad\qquad\qquad\qquad $K\leftarrow (k,\,
\frac{a'*L_0-k-i_0}{n}, \,\cdots,
\,\frac{a'*L_{t-1}-k-i_{t-1}}{n})$;\\
14. \qquad\qquad\qquad \verb"End If";\\
15. \qquad\qquad \verb"End If";\\
16. \qquad \verb"End For";\\
17. \verb"End For";\\
18. \verb"Output" $K$;

\section{Application to second order nonlinearity estimation of cubic Boolean functions}\label{Application}
Following Lemma describes lower bounds of the second-order
nonlinearities of cubic Boolean functions by the dimensions of root
sets of linearized polynomials.

\begin{lemma}\label{lemfunda}
Let $f$  be any cubic Boolean function. Define $Q_{f}:=\{a\in
\GF{n}|nl(D_af)\neq0\}$. Let us suppose that for every element $a
\in Q_{f}$, the dimension of the kernel of the derivative $D_af$
(or, equivalently, its quadratic part) of $f$ at $a$ is not greater
than $t$, where $t\geq 0$  is some fixed integer. Then
\[
nl_2(f)\geq
2^{n-1}-\frac{1}{2}\sqrt{2^{2n}-2|Q_{f}|(2^{n-1}-2^{\lfloor\frac{n+t}{2}\rfloor-1})}.
\]
\end{lemma}
\begin{proof}
This is an immediate corollary from \eqref{Wal-nonli}, Lemma
\ref{parity}, Lemma \ref{kerdim} and Theorem \ref{r-nonli-r-1}.
\end{proof}

Following theorem gives the most precise estimation for lower bound
of the second-order nonlinearity of any cubic Boolean function no
possessing affine derivatives, including the special form
$G_\mu=Tr(\sum_{l=1}^{m}\mu_l x^{d_l})$, where
$d_l=2^{i_l\gamma+j_l\gamma+1}$.

\begin{theorem}\label{maintheorem}
Let  $F_\mu=Tr(\sum_{l=1}^{m}\mu_l x^{d_l})$, where $\mu_l\in
\GF{n}^*$ and $d_l=2^{i_l+j_l+1}$, $i_l>j_l>0$, be any cubic Boolean
function. Define $Q_{F_\mu}:=\{a\in \GF{n}|nl(D_aF_\mu)\neq0\}$. Let
$\psi_a(x)=Tr(\sum_{i=1}^{n-1}c_{i,a}x^{2^i+1}), c_{i,a}\in \GF{n}$,
be the quadratic part of the derivative of $F_\mu$ at $a\in
Q_{F_\mu}$.

Let  $\Delta=\{\,i\,|\,\exists a\in Q_{F_\mu}, c_{i,a}\neq 0, 1\leq
i \leq n-1\,\}\cup \{\,-i\,|\,\exists a\in Q_{F_\mu}, c_{i,a}\neq 0,
1\leq i \leq n-1\,\}=\{i_0,i_1,\cdots,i_{t-1}\}$ and
$U=\{\,K=(k,k_0,k_1,\cdots,k_{t-1})\in \mathbb{Z}^{t+1}\,|\,\forall
j\in \overline{\{0,t-1\}}, i_j+k+k_jn\geq0\,\}$. For $K\in U$, let
us define following quantities sequentially:
$T_K=\gcd(\{i_j+k+k_jn|j\in\overline{\{0,t-1\}}\})$, $S_K=T_K/
gg(T_K,n)$,
$V_K=\max_{j\in\overline{\{0,t-1\}}}\{\frac{i_j+k+k_jn}{S_K}\}$ and
$V=\min_{K\in U}V_K$.

Then
\begin{equation}
nl_2(F_\mu)\geq
2^{n-1}-\frac{1}{2}\sqrt{2^{2n}-2|Q_{F_\mu}|(2^{n-1}-2^{\lfloor\frac{n+V}{2}\rfloor-1})},
\end{equation}
 and this estimation is
at least as much precise as ones in Theorem 9 and 10.

In particular, if $|Q_{F_\mu}|=2^n-1$, i.e. for every $a\in
\GF{n}^*$, $D_aF_\mu$ is not affine, then it holds
\begin{equation}
nl_2(F_\mu)\geq
2^{n-1}-\frac{1}{2}\sqrt{2^n+(2^n-1)2^{\lfloor\frac{n+V}{2}\rfloor}}.
\end{equation}
\end{theorem}
\begin{proof}

From Lemma \ref{lemfunda}, one can see that a lower bound of
second-order nonlinearity of $F_\mu$  is obtained from a upper bound
for dimension of kernel of
$\psi_a(x)=Tr(\sum_{i=1}^{n-1}c_{i,a}x^{2^i+1})$, the quadratic part
of the derivative $D_aF_\mu$. The kernel $\varepsilon_{F_\mu,a}$ of
$\psi_a(x)$ is given as the set of $x\in \GF{n}$ such that for any
$y\in \GF{n}$
$B_a(x,y)=\psi_a(x)+\psi_a(y)+\psi_a(x+y)=Tr(y\sum_{i=1}^{n-1}(c_{i,a}x^{2^i}+(c_{i,a}x)^{2^{-i}})=0$
  , i.e. the root set of linearized polynomial
\begin{equation}\label{polarform}
\sum_{i=1}^{n-1}(c_{i,a}x^{2^i}+(c_{i,a}x)^{2^{-i}}).
\end{equation}
   Applications of  Theorem \ref{corrootnum} and Lemma \ref{lemfunda} give the main assertion of the theorem.

Let us compare the lower bound estimation given in Theorem 11 with
ones of Li, Hu and Gao. First remark that by the Note we made in
Section 2 we can suppose  $t\leq \lfloor\frac{n}{2}\rfloor$  and
therefore the cases \textcircled{1} and \textcircled{2} of Theorem
\ref{Li-Hu-Gao-F} can be excluded from consideration. The Li-Hu-Gao
estimation is obtained as a special case of our discussion: Let
$t=\max\{\,i\in\Delta\,|\,i>0\,\}$,
$s=\min\{\,i\in\Delta\,|\,i>0\,\}$, $t_1=\min\{\,i\in\Delta\,|\,i>0,
i\neq t\,\}$, using $\Delta$ introduced by us. Taking two integer
vectors $K_1=\{t,0,\cdots,0\}$ ( $|\Delta|$ 0's ),
$K_2=\{-s,0,\cdots,0,1,\cdots,1\}$ ($\frac{|\Delta|}{2}$ 1's and
$\frac{|\Delta|}{2}$ 0's) for \textcircled{3} (case $n>2t$) of
Theorem \ref{Li-Hu-Gao-F} and taking $K_1=\{t,0,\cdots,0\}$(
$|\Delta|$ 0's ),
$K_2=\{-s,0,\cdots,0,-1,0,\cdots,0,1,\cdots,1\}$(the numbers of 0's
and 1's are $\frac{|\Delta|}{2}-1$, $\frac{|\Delta|}{2}$,
respectively and the place number of -1's is  $k_t$) for
\textcircled{4} (case $n=2t$) of Theorem \ref{Li-Hu-Gao-F}, then
letting $V_0=\min\{V_1,V_2\}$, give
\[
nl_2(F_\mu)\geq
2^{n-1}-\frac{1}{2}\sqrt{2^n+(2^n-1)2^{\lfloor\frac{n+V_0}{2}\rfloor}}.
\]
Obviously $V_0\geq V$, therefore our estimation would be at least as
much precise as ones given by Li-Hu-Gao.  Comparison with Theorem
\ref{Li-Hu-Gao-G} is also similar.
\end{proof}

 Finally, we note that  an assumption $c_{i,a}=0$ when $i>\lfloor\frac{n}{2}\rfloor$
can be made in the formulation of Theorem \ref{maintheorem}.

\section{Examples and comparisons}\label{Example}
As shown in below examples, for almost all cases, our estimation
would be more precise than ones of Li, Hu and Gao \cite{8}.

\begin{example}(Example 1 of \cite{8})
  Let $F_{\mu}=f_\mu=Tr(\mu x^{2^i+2^j+1})$. For every  $u\in \GF{n}^*$,
  the quadratic part of the derivative of $F_{\mu}$ is represented as
   $h_u(x)=Tr(\lambda_u^{2^{n-j}}x^{2^{i-j}+1}+\lambda_u x^{2^i+1}+\lambda_u x^{2^j+1})$
   for some $\lambda_u\in \GF{n}^*$.

\begin{enumerate}

\item $n=20, i=9, j=5$\\
Theorem \ref{Gode-Gango-F} says
\[
nl_2(f_{\mu})\geq
2^{19}-\frac{1}{2}\sqrt{2^{20}+(2^{20}-1)2^{19}}\approx 153561,
\]
 and Theorem \ref{Li-Hu-Gao-F}
says (in this case  $s=i-j=4, t=i=9$, and since $n>2t$ is an even,
we can set $p=\min\{12,18\}=12$ by \textcircled{4} of Theorem
\ref{Li-Hu-Gao-F})
\[
nl_2(f_\mu)\geq
2^{19}-\frac{1}{2}\sqrt{2^{20}+(2^{20}-1)2^{16}}\approx 393216.
\]

Now we will apply Theorem \ref{maintheorem} to this case. By
definition, $\Delta=\{i,j,i-j,-i,-j,j-i\}=\{9,5,4,-9,-5,-4\}$. For
$K=\{-5,2,0,5,4,6,1\}$,
$T_K=\gcd(9-5+40,5-5,4-5+100,-9-5+80,-5-5+120,-4-5+20)=\gcd(44,0,99,66,110,11)=11$,
$S_K=T_K=11$. Thus  $V\leq V_K=\max\{4,0,9,6,10,1\}=10$ and by
Theorem \ref{maintheorem} we have
\[
nl_2(f_\mu)\geq
2^{19}-\frac{1}{2}\sqrt{2^{20}+(2^{20}-1)2^{15}}\approx 431605.
\]

\item $n=19, i=9, j=5$

Theorem \ref{Gode-Gango-F} asserts
 \[
nl_2(f_\mu)\geq
2^{18}-\frac{1}{2}\sqrt{2^{19}+(2^{19}-1)2^{18}}\approx 76781.
\]
Theorem \ref{Li-Hu-Gao-F} gives (in this case  $s=i-j=4,t=i=9$ and
since $n>2t$  is an odd, we can set $q=\min\{11,18\}=11$ by
\textcircled{4} of Theorem \ref{Li-Hu-Gao-F})
 \[
nl_2(f_\mu)\geq
2^{18}-\frac{1}{2}\sqrt{2^{19}+(2^{19}-1)2^{15}}\approx 196608.
\]

On the other hand, the application of our Theorem \ref{maintheorem}
can improve these estimations as follows. By definition,
$\Delta=\{i,j,i-j,-i,-j,j-i\}=\{9,5,4,-9,-5,-4\}$. For
$K=\{-4,0,1,0,2,1,2\}$,
$T_K=\gcd(9-4,5-4+19,4-4,-9-4+38,-5-4+19,-4-4+38)=\gcd(5,20,0,25,10,30)=5$,
$S_K=T_K=5$. Thus  $V\leq V_K=\max\{1,4,0,5,2,6\}=6$ and Theorem
\ref{maintheorem} shows
\[
nl_2(f_\mu)\geq
2^{18}-\frac{1}{2}\sqrt{2^{19}+(2^{19}-1)2^{12}}\approx 238971.
\]
\end{enumerate}
\end{example}

The lower bound given by Theorem \ref{maintheorem} also improves the
Li-Hu-Gao estimation (Theorem \ref{Li-Hu-Gao-G}) for Boolean
functions $G_\mu$.

\begin{example}
Let $G_\mu(x)=Tr(\mu x^{2^{i\gamma}+2^{j\gamma}+1})$. The quadratic
part of the derivative of $G_\mu$ at  $u\in \GF{n}^*$ is represented
as
$h_u(x)=Tr(\lambda_u^{2^{n-j\gamma}}x^{2^{i\gamma-j\gamma}+1}+\lambda_u
x^{2^{i\gamma}+1}+\lambda_u x^{2^{j\gamma}+1})$ for some
$\lambda_u\in \GF{n}^*$.

\begin{enumerate}
\item $n=20,i=9,j=5,\gamma=2$.

Since $n\neq (i+j)\gamma, n\neq(2i-j)\gamma$, by Theorem 2 of
\cite{7} $G_\mu$ has no affine derivative. Due to $n>2i$, by Theorem
\ref{Li-Hu-Gao-G} we have

\[
nl_2(G_\mu)\geq
2^{19}-\frac{1}{2}\sqrt{2^{20}+(2^{20}-1)2^{19}}\approx 153561.
\]

At this time, let us use Theorem \ref{maintheorem} to estimate
$nl_2(G_\mu)$. By definition,
$\Delta=\{2i,2j,2i-2j,-2i,-2j,2j-2i\}=\{18,10,8,-18,-10,-8\}$. For
$K=\{8,-1,0,1,2,1,0\}$,
$T_K=\gcd(18+8-20,10+8,8+8+20,-18+8+40,-10+8+20,-8+8)=\gcd(6,18,36,30,18,0)=6$,
$gg(T_K,n)=2$, $S_K=T_K/2=3$. Thus $V\leq
V_K=\max\{2,6,12,10,6,0\}=12$ and Theorem \ref{maintheorem} gives an
improved estimation
\[
nl_2(G_\mu)\geq
2^{19}-\frac{1}{2}\sqrt{2^{20}+(2^{20}-1)2^{16}}\approx 393216.
\]

\item $n=19,i=9,j=5,\gamma=2$.

Since $n\neq (i+j)\gamma, n\neq(2i-j)\gamma$,  $G_\mu$ has no affine
derivative. Due to $n>2i$, Theorem \ref{Li-Hu-Gao-G} says

\[
nl_2(G_\mu)\geq
2^{18}-\frac{1}{2}\sqrt{2^{19}+(2^{19}-1)2^{18}}\approx 76781.
\]

Next, we will estimate  $nl_2(G_\mu)$ by using Theorem
\ref{maintheorem}. By definition,
$\Delta=\{2i,2j,2i-2j,-2i,-2j,2j-2i\}=\{18,10,8,-18,-10,-8\}$. For
$K=\{8,1,0,2,1,2,0\}$,
$T_K=\gcd(18+8+19,10+8,8+8+38,-18+8+19,-10+8+38,-8+8)=\gcd(45,18,54,9,36,0)=9$,
$S_K=T_K=9$. Thus $V\leq V_K=\max\{5,2,6,1,4,0\}=6$ and Theorem
\ref{maintheorem} proves the improved estimation
\[
nl_2(G_\mu)\geq
2^{18}-\frac{1}{2}\sqrt{2^{19}+(2^{19}-1)2^{12}}\approx 238971.
\]
\end{enumerate}
\end{example}

\begin{example}
For  $f_\mu$, the case of  $n=i+j,n\neq 2i-j$ is treated as
Corollary 4 in \cite{8}. Apply Theorem \ref{maintheorem} to this
case: $\Delta=\{i,j,i-j,-i,-j,j-i\}$.
 For  $K=\{2j,-1,0,-1,1,0,1\}$,  $T_K=j, S_K=j/ gg(j,n)$.
Thus  $V\leq V_K=4gg(j,n)$ and Theorem \ref{maintheorem} indicates
\begin{equation}\label{i+j}
nl_2(f_\mu)\geq
2^{n-1}-\frac{1}{2}\sqrt{2^n++(2^n-1)2^{\lfloor\frac{n+4gg(j,n)}{2}\rfloor}}.
\end{equation}
And in particular, if $\gcd(j,n)=1$ (so $gg(j,n)=1$), then
\[
2^{n-1}-\frac{1}{2}\sqrt{2^n+(2^n-1)2^{\lfloor\frac{n+4}{2}\rfloor}}.
\]
This lower bound is better than ones (with complicated
representations) given by Corollary 4 of \cite{8}. In fact, since
\[
 f_\mu=g_{\mu^p},
\]
this is not other than Corollary 5 of \cite{8} applied to $g_\mu$,
or, Theorem \ref{Gode-Gango-G}. How to improve this lower bound is
discussed in Section 7.
\end{example}

The exact values for the maximum second-order nonlinearity that a
$n-$variable Boolean function can achieve (i.e. the covering radius
of $RM(2,n)$) are known only for $3\leq n\leq 6$ \cite{14}; its
value is 1, 2, 6 an 18 respectively. It is conjectured in \cite{15}
that the exact value of the maximum second-order nonlinearity is
attained by a coset of $RM(2,n)$ in $RM(3,n)$ (i.e. by a cubic
function). Following examples also confirm this conjecture.
\begin{example}\label{exam1}
For the modified-Welch Boolean function $f_{welch'}=Tr(x^{2^t+3})$,
$t=\frac{n+1}{2}$, $n$ odd, Carlet's lower bound (Proposition 5 of
\cite{5}) states
\[
nl_2(f_{welch'})\geq
2^{n-1}-\frac{1}{2}\sqrt{2^n+(2^n-1)2^{\frac{n+3}{2}}}.
\]

 For odd $n>1$ (i.e. $n=3$) smaller than 5, this lower
bound becomes zero (the approximation also becomes equality) and
therefore non-meaningful.

But Theorem \ref{maintheorem} gives a meaningful lower bound as
follows: We have
$D_af_{welch'}(x)=Tr(ax^{2^t+2}+a^2x^{2^t+1}+a^{2^t}x^3)+l(x)=Tr(a^4x^3)+l(x)$
where $l$ is affine. Therefore $\Delta=\{1,-1\}$. Take
$K=\{1,0,0\}$. Then $V_K=1$. In fact, the kernel of the quadratic
Boolean function $Tr(a^4x^3)$ is $\{0, a\}$ when $a\neq 0$, and
therefore has the exact dimension 1. Hence for $n=3$ we have
\[
nl_2(f_{welch'})\geq
2^{n-1}-\frac{1}{2}\sqrt{2^n+(2^n-1)2^{\frac{n+1}{2}}}=1,
\]
that is, $nl_2(f_{welch'})=1$ over $\GF{3}$.
\end{example}

\begin{example}\label{exam2}
For $n=4$, consider the function $f=Tr(x^{2^3+2^2+1})$. Note
$f=Tr((x^{2^3+2^2+1})^4)=Tr(x^{2^2+2+1})$. At $a\in \GF{4}$, it has
derivative
$D_af=Tr(ax^{2^3+2^2}+a^{2^2}x^{2^3+1}+a^{2^3}x^{2^2+1})=Tr((a^2+a^{2^2})x^{2^3+1}+a^{2^3}x^{2^2+1})$.
If $a=0$ or $a=1$, then $D_af=0$ and $Q_f=\GF{4}\setminus \{0,1\}$.
For $a\neq 0,1$, We have $\Delta=\{3,2,-2,-3\}$, and taking
$K=\{3,-1,-1, 0, 0\}$, we get $V\leq V_K=2$. Following discussion
shows really $V=2$: The kernel $\varepsilon_{f,a}$ of $D_af$ is the
null space of
\begin{align*}
&(a^2+a^{2^2})x^{2^3}+((a^2+a^{2^2})x)^{2^{-3}}+a^{2^3}x^{2^2}+(a^{2^3}x)^{2^{-2}}\\
&=(a+a^2)^2x^8+(a+a^2)^4x^2+(a^8+a^2)x^4\\
&=[(a+a^2)x^4+(a^4+a)x^2+(a^4+a^2)x]^2\\
&=[(a+a^2)(x^2+x)^2+(a^4+a^2)(x^2+x)]^2\\
&=(a+a^2)^2(x^2+x)^2(x^2+x+a^2+a)^2,
\end{align*}
i.e. $\varepsilon_{f,a}=\{0,1,a,1+a\}$ and $V=r_a=2$.

By using Theorem \ref{maintheorem}, we have
\[
nl_2(f)\geq
2^{n-1}-\frac{1}{2}\sqrt{2^{2n}-2(2^n-2)(2^{n-1}-2^{\frac{n}{2}})}=2^3-\frac{1}{2}\sqrt{2^8-2\times14\times4}=2,
\]
that is, $nl_2(f)=2$ over $\GF{4}$.

\end{example}

\begin{example}\label{manypaper}The second-order nonlinearity of $g_\lambda=Tr(\lambda x^{2^{2r}+2^r+1})$ over $\GF{n}$ with $n=sr$
has been studied for $s=3,4,5,6$ by independent papers:
\begin{enumerate}
\item Singh \cite{17} discussed the case $s=3$.  Li-Hu-Gao \cite{8} also discussed this case
(Corollary 3 of \cite{8}).
\item Sun and Wu \cite{16} discussed the case $s=4$.
\item Gangopadhyay and Garg \cite{11} discussed the case $s=5$.
\item Gangopadhyay, Sarkar and Telang \cite{6} discussed the case $s=6$.
\end{enumerate}

The lower bounds proved by all these works can be shown or even
improved by corollaries of Theorem \ref{maintheorem}: Remind
$\Delta=\{2r,r,-r,-2r\}$.
\begin{enumerate}
\item For $n=3r$, by taking $K=\{2r, -1, -1, 0, 0\}$, $V\leq V_K=r$.
\begin{equation}
nl_2(g_\lambda)\geq 2^{n-1}-\frac{1}{2}\sqrt{2^n+(2^n-1)2^{2r}}.
\end{equation}
\item For $n=4r$, by taking $K=\{3r, -1, -1, 0, 0\}$, $V\leq V_K=2r$.
\begin{equation}
nl_2(g_\lambda)\geq 2^{n-1}-\frac{1}{2}\sqrt{2^n+(2^n-1)2^{3r}}.
\end{equation}
\item For $n=5r$, by taking $K=\{4r, -1, -1, 0, 0\}$, $V\leq V_K=3r$.
\begin{equation}
nl_2(g_\lambda)\geq 2^{n-1}-\frac{1}{2}\sqrt{2^n+(2^n-1)2^{4r}}.
\end{equation}
\item For $n=6r$, by taking $K=\{2r, 0, 0, 0, 0\}$, $V\leq V_K=4r$.
\begin{equation}
nl_2(g_\lambda)\geq 2^{n-1}-\frac{1}{2}\sqrt{2^n+(2^n-1)2^{5r}}.
\end{equation}
\end{enumerate}
\end{example}

 Furthermore, while for $s\geq 8$ the minimum $V$ search program gives only $V\leq
4r$ which is trivial, for $n=7r$ a better result is shown: One can
choose an integer $k$ such that $\gcd(n, 7k+4)=1$. Then, by taking
$K=\{6r, 2k, -1, 3k+1, k\}$ we have $V\leq
 V_K=\max\{(14k+8)r/(7k+4), 0, (21k+12)r/(7k+4), (7k+4)r/(7k+4)\}=3r$ and thus
 a novel result:
\begin{corollary}\label{manypaper_new}
If $n=7r$, then
\begin{equation}
nl_2(g_\lambda)\geq 2^{n-1}-\frac{1}{2}\sqrt{2^n+(2^n-1)2^{5r}}.
\end{equation}
\end{corollary}

\section{Towards better lower bounding}\label{better}
In this section, it is shown that \eqref{Mesnager} based on
studying the distribution of $\{r_a,a\in \GF{n}^*\}$ would lead to
better lower bound on the second-order nonlinearity.
\subsection{Specific Case}\label{special sec} Consider the cubic Boolean function
$f_7=Tr(x^7)=Tr(x^{2^2+2+1})$. This function is a special case (with
$r=1, \mu=1$) of the wider Boolean function family $g_{\mu}=Tr(\mu
x^{2^{2r}+2^r+1})$ which will be considered in the next subsection.
It was known that when $n=4r$, $g_{\mu}$ is highly nonlinear
permutation \cite{19}, and has differential uniformity of four
\cite{18}, and thus the same resistance to both differential and
linear attacks as the inverse function.

In Example \ref{exam1} and Example \ref{exam2}, we considered that
for the cases $n=3$ and $n=4$ this Boolean function achieves the
maximum second-order nonlinearity. For $n\geq 5$ Theorem
\ref{maintheorem} can give only the same lower bound as Theorem
\ref{Li-Hu-Gao-G} because  $V=3$ for $n=5,7$ and $V=4$ for other
values of $n$. In this section, we show that \eqref{Mesnager} based
on studying the distribution of $\{r_a,a\in \GF{n}^*\}$ leads to a
better lower bounding for $nl_2(f_7)$.

The quadratic part of derivative $D_af_7$ of $f_7$ at $a\in \GF{n}$
is $Tr(a^4x^3+a^2x^5+ax^6)$, and $\varepsilon_{f_7,a}$ is the root
set of the linearized polynomial
$a^4x^2+(a^4x)^{2^{-1}}+a^2x^4+(a^2x)^{2^{-2}}+(ax^2)^{2^{-2}}+(ax^4)^{2^{-1}}$
(refer to \eqref{polarform}). We have
\begin{align*}
&a^4x^2+(a^4x)^{2^{-1}}+a^2x^4+(a^2x)^{2^{-2}}+(ax^2)^{2^{-2}}+(ax^4)^{2^{-1}}=0\\
&\Longleftrightarrow a^{16}x^8+a^8x^2+a^{8}x^{16}+a^2x+ax^2+a^2x^8=0\\
&\Longleftrightarrow a^{8}x^{16}+(a^{16}+a^2)x^8+(a^8+a)x^2+a^2x=0\\
&\Longleftrightarrow (ax)^{8}(a+x)^8+(ax)^2(a^3+x^3)^2+ax(a+x)=0\\
&\Longleftrightarrow (ax)^{8}(a+x)^8+(ax)^2(a+x)^2(a^2+ax+x^2)^2+ax(a+x)=0\\
&\Longleftrightarrow ax(a+x)\left[(ax)^{7}(a+x)^7+ax(a+x)(a^2+ax+x^2)^2+1\right]=0\\
&\Longleftrightarrow ax(a+x)\left[a^{7}(ax+x^2)^7+a(ax+x^2)(a^4+(ax+x^2)^2)+1\right]=0\\
&\Longleftrightarrow ax(a+x)\left[a^{5}(ax+x^2)(a^{2}(ax+x^2)^6+1)+a(ax+x^2)^3+1\right]=0\\
&\Longleftrightarrow ax(a+x)\left[a^{5}(ax+x^2)(a(ax+x^2)^3+1)^2+(a(ax+x^2)^3+1)\right]=0\\
&\Longleftrightarrow ax(a+x)\left[a(ax+x^2)^3+1\right]\left[a^{5}(ax+x^2)(a(ax+x^2)^3+1)+1\right]=0\\
&\Longleftrightarrow \left(ax+x^2\right)\cdot\left[(ax+x^2)^3+\frac{1}{a}\right]\cdot\left[(ax+x^2)^4+\frac{1}{a}(ax+x^2)+\frac{1}{a^6}\right]=0.\\
\end{align*}
Consequently,  $\varepsilon_{f_7,a}=K_{a,1}\cup K_{a,2}\cup
K_{a,3}$, where $K_{a,1}=\{x\in\GF{n}|ax+x^2=0\}=\{0,a\}$,
$K_{a,2}=\{x\in\GF{n}|(ax+x^2)^3=\frac{1}{a}\}$,
$K_{a,3}=\{x\in\GF{n}|(ax+x^2)^4+\frac{1}{a}(ax+x^2)=\frac{1}{a^6}\}$.
Note the polynomial
$\left(ax+x^2\right)\cdot\left[(ax+x^2)^3+\frac{1}{a}\right]\cdot\left[(ax+x^2)^4+\frac{1}{a}(ax+x^2)+\frac{1}{a^6}\right]$
is separable and so $K_1,K_2,K_3$ are disjoint each one to another.

Now, we will consider $|K_{a,2}|$ and $|K_{a,3}|$. First, note that
\begin{equation}\label{k2k3}
|K_{a,2}|\leq 6, |K_{a,3}|\leq 8
\end{equation}
 and that Lemma
\ref{parity} let us know that
\begin{equation}\label{k2+k3}
|K_{a,2}|+|K_{a,3}|=
\begin{cases}
2 \text{ or } 14, &\text{ if $n$ is even;}\\
0 \text{ or } 6, &\text{ if $n$ is odd.}
\end{cases}
\end{equation}

Then, from an easy consideration, one can see:
$K_{a,2}\neq\emptyset$ if{f} $a$ is a cubic element in $\GF{n}$ and
$Tr(\frac{1}{a^{2}b})=0$ for a cubic root $b$ of $a$, i.e. such as
$b^3=a$.

There are two cases to consider:

\begin{enumerate}
\item If $n$ is even, then the 3-th powering is a three-to-one mapping of $\GF{n}^*$, and so there are $\frac{2(2^n-1)}{3}$
$a$'s with $K_{a,2}=\emptyset$ (in this case, by \eqref{k2k3} and
\eqref{k2+k3} it must be $|K_{a,3}|=2$). For remained
$\frac{(2^n-1)}{3}$ $a$'s,
\begin{equation*}
|K_{a,2}|=
\begin{cases}
6,&\text{if $Tr(\frac{\zeta^i}{a^{7/3}})=0$ for all $0\leq i<3$;}\\
2,&\text{otherwise,}
\end{cases}
\end{equation*}
\begin{equation*}
|K_{a,3}|=
\begin{cases}
8,&\text{if $Tr(\frac{\zeta^i}{a^{7/3}})=0$ for all $0\leq i<3$;}\\
0,&\text{otherwise.}
\end{cases}
\end{equation*}

After all,  for even $n$, denoting
\[
\Psi_e=\{a\in \GF{n}^*|a \text{ is a cubic
 and } Tr(\frac{1}{a^2b})=0 \text{ for every cubic root } b \text{ of }
a\},
\]
we have
\begin{align*}
&\{a\in\GF{n}^*|r_a=4\}=\Psi_e,\\
&\{a\in\GF{n}^*|r_a=2\}=\GF{n}^*\setminus \Psi_e.
\end{align*}

It should be stressed that $|\Psi_e|\leq\frac{(2^n-1)}{3}$. By
\eqref{Mesnager}, we get
\begin{align*}
nl_2(f_7)&\geq \max\left(2^{n-2}-2^{\frac{n-2}{2}},
2^{n-1}-\frac{1}{2}\sqrt{2^n+|\Psi_e|2^{\frac{n+4}{2}}+(2^n-1-|\Psi_e|)2^{\frac{n+2}{2}}}\right)\\
&=\max\left(2^{n-2}-2^{\frac{n-2}{2}},
2^{n-1}-\frac{1}{2}\sqrt{2^n+|\Psi_e|2^{\frac{n+2}{2}}+(2^n-1)2^{\frac{n+2}{2}}}\right)\\
&\geq \max\left(2^{n-2}-2^{\frac{n-2}{2}},
2^{n-1}-\frac{1}{2}\sqrt{2^n+\frac{(2^n-1)}{3}2^{\frac{n+2}{2}}+(2^n-1)2^{\frac{n+2}{2}}}\right),
\end{align*}
i.e. for even $n\geq 6$ we have
\begin{equation}\label{f7-neven}
nl_2(f_7)\geq
2^{n-1}-\frac{1}{2}\sqrt{2^n+\frac{8}{3}\cdot2^{\frac{3n}{2}}-\frac{8}{3}2^{\frac{n}{2}}}.
\end{equation}

If $3\dag n$ and therefore the 7-th powering is a permutation of
$\GF{n}$, then for any cubics $a\neq a'\in \GF{n}^*$, when $b^3=a$
and $b'^3=a'$, one has $\frac{1}{a^2b}\neq \frac{1}{a'^2b'}$,
because third powering to the both side of $\frac{1}{a^2b}=
\frac{1}{a'^2b'}$ leads to $a^7=a'^7$ i.e. $a=a'$ i.e. a
contradiction. Thus, when $a$ takes all cubics of $\GF{n}^*$ and $b$
takes all three cubic roots of $a$, $\frac{1}{a^2b}$ takes all
$2^n-1$ elements in $\GF{n}^*$. Since in $\GF{n}^*$ there are
$2^{n-1}-1$ elements with absolute trace 0, it follows that
$|\Psi_e|\leq\frac{(2^{n-1}-2)}{3}$. Hence,
\begin{align*}
nl_2(f_7)&\geq
2^{n-1}-\frac{1}{2}\sqrt{2^n+|\Psi_e|2^{\frac{n+2}{2}}+(2^n-1)2^{\frac{n+2}{2}}}\\
&\geq
2^{n-1}-\frac{1}{2}\sqrt{2^n+\frac{(2^{n-1}-2)}{3}2^{\frac{n+2}{2}}+(2^n-1)2^{\frac{n+2}{2}}},
\end{align*}
i.e. when $n\equiv 2,4 \mod 6$, we have
\begin{equation}\label{f7-neven}
nl_2(f_7)\geq
2^{n-1}-\frac{1}{2}\sqrt{2^n+\frac{7}{3}\cdot2^{\frac{3n}{2}}-\frac{10}{3}\cdot2^{\frac{n}{2}}}.
\end{equation}

\item If $n$ is odd, then the 3-th power mapping is a permutation of
$\GF{n}$ and therefore we have:
\[
K_{a,2}\neq \emptyset\quad \text{if{f}} \quad
Tr(\frac{1}{a^{7/3}})=0 \quad \text{if{f}} \quad  |K_{a,2}|=2.
\]
The 7-th power mapping in $\GF{n}^*$ is injective if $3\dag n$ and
eight-to-one if $3|n$. Therefore, the number of $a(\neq0)$'s with
$|K_{a,2}|=2$ is $2^{n-1}-1$ if $2,3\dag n$ (i.e. $n\equiv \pm 1
\mod 6$) and $2^n-wt(f_7)-1$ if $2\dag n$ and $3|n$ (i.e. $n\equiv 3
\mod 6$). Furthermore, with regard to \eqref{k2k3} and
\eqref{k2+k3}, if $|K_{a,2}|=2$ then $|K_{a,3}|=4$.

On the other hand, it can not happen $|K_{a,3}|=6$. In fact,
$|K_{a,3}|=6$ means that the degree-4 equation
$T^4+\frac{1}{a}T+\frac{1}{a^6}=0$ with $T=ax+x^2$ has exactly 4
solutions $T_1, T_2, T_3, T_4$ in $\GF{n}$ such that
$Tr(\frac{T_1}{a^2})=Tr(\frac{T_2}{a^2})=Tr(\frac{T_3}{a^2})=0$ and
 $Tr(\frac{T_4}{a^2})=1$, which can not
happen because $T_1+T_2+T_3+T_4=0$. Hence, if $|K_{a,2}|=0$ then
$|K_{a,3}|=0$.

After all, for odd $n$, denoting
\[
\Psi_o=\{a\in \GF{n}^*| Tr(\frac{1}{a^{7/3}})=0\},
\]
we have
\begin{align*}
&\{a\in\GF{n}^*|r_a=3\}=\Psi_o,\\
&\{a\in\GF{n}^*|r_a=1\}=\GF{n}^*\setminus \Psi_o.
\end{align*}
Here, if $n\equiv \pm 1 \mod 6$ then $|\Psi_o|=2^{n-1}-1$, and if
$n\equiv 3 \mod 6$ then $|\Psi_o|=2^n-wt(f_7)-1$.
\end{enumerate}

By \eqref{Mesnager}, we get
\begin{align*}
nl_2(f_7)&\geq \max\left(2^{n-2}-2^{\frac{n-3}{2}},
2^{n-1}-\frac{1}{2}\sqrt{2^n+|\Psi_o|2^{\frac{n+3}{2}}+(2^n-1-|\Psi_o|)2^{\frac{n+1}{2}}}\right)\\
&=\max\left(2^{n-2}-2^{\frac{n-3}{2}},
2^{n-1}-\frac{1}{2}\sqrt{2^n+|\Psi_o|2^{\frac{n+1}{2}}+(2^n-1)2^{\frac{n+1}{2}}}\right).
\end{align*}
If $n\equiv \pm 1 \mod 6$, then this gives
\begin{align*}
nl_2(f_7)&\geq \max\left(2^{n-2}-2^{\frac{n-3}{2}},
2^{n-1}-\frac{1}{2}\sqrt{2^n+(2^{n-1}-1)2^{\frac{n+1}{2}}+(2^n-1)2^{\frac{n+1}{2}}}\right),
\end{align*}
i.e. for $n$ such as $n\equiv \pm 1 \mod 6$ and $n\geq 5$,
\begin{equation}\label{f7-n15}
nl_2(f_7)\geq
2^{n-1}-\frac{1}{2}\sqrt{2^n+3\cdot2^{\frac{3n+1}{2}}-2^{\frac{n+3}{2}}}.
\end{equation}
When $n\equiv 3 \mod 6$ and $n\geq 5$, we obtain
\begin{equation}\label{f7-n3}
nl_2(f_7)\geq
2^{n-1}-\frac{1}{2}\sqrt{2^n+(2^n-1)2^{\frac{n+3}{2}}-wt(f_7)2^{\frac{n+1}{2}}}.
\end{equation}

\subsection{Generalization to $g_\mu$ with $\gcd (n,r)=1$} An improved lower bound on second-nonlinearity of the cubic Boolean function
$g_{\mu}=Tr(\mu x^{2^{2r}+2^{r}+1})$, where $\mu \in \GF{n}$ and
$\gcd (n,r)=1$, is derived in this subsection, which can be seen as
a generalization of Subsection \ref{special sec}.

Denote $p=2^r$. The quadratic part of derivative $D_ag_\mu$ of
$g_\mu$ at $a\in \GF{n}$ is $Tr(\mu ax^{p^2+p}+\mu a^px^{p^2+1}+\mu
a^{p^2}x^{p+1})=Tr(\mu a^px^{p^2+1}+((\mu a)^{\frac{1}{p}}+\mu
a^{p^2})x^{p+1})$, and $\varepsilon_{g_\mu,a}$ is the root set of
the linearized polynomial
\[
L_{\mu,a}(x)=\mu a^px^{p^2}+((\mu a)^{\frac{1}{p}}+\mu
a^{p^2})x^{p}+(\mu a^px)^{\frac{1}{p^2}}+(((\mu a)^{\frac{1}{p}}+\mu
a^{p^2})x)^{\frac{1}{p}}
\]
(refer to \eqref{polarform}). We have
\begin{align*}
&\mu a^px^{p^2}+((\mu a)^{\frac{1}{p}}+\mu a^{p^2})x^{p}+(\mu
a^px)^{\frac{1}{p^2}}+(((\mu a)^{\frac{1}{p}}+\mu
a^{p^2})x)^{\frac{1}{p}}=0\\
&\Longleftrightarrow \mu^{p^2} a^{p^3}x^{p^4}+(\mu^p a^p+\mu^{p^2}
a^{p^4})x^{p^3}+\mu a^px+(\mu a+\mu^p a^{p^3})x^p=0,
\end{align*}
i.e.
\begin{equation}\label{maineq}
\mu^{p^2} (ax^p+a^px)^{p^3}+\mu^p( ax^{p^2}+
a^{p^2}x)^p+\mu(ax^p+a^px)=0.
\end{equation}

 Now, we let $z:=ax^p+a^px$. Then,
$x^p=\frac{z+a^px}{a}$ and $ x^{p^2}=\frac{z^p+a^{p^2}x^p}{a^p}$,
and
\[
ax^{p^2}+
a^{p^2}x=\frac{z^p+a^{p^2}x^p}{a^{p-1}}+a^{p^2}x=\frac{z^p+a^{p^2}x^p+a^{p^2+p-1}x}{a^{p-1}}=\frac{z^p+a^{p^2-1}z}{a^{p-1}}.
\]

Therefore, the above equation becomes
\begin{align*}
\mu^{p^2}
z^{p^3}+\mu^p\frac{z^{p^2}+a^{p(p^2-1)}z^p}{a^{p(p-1)}}+\mu z=0,
\end{align*}
or, equivalently
\begin{align*}
&\mu^{p^2}a^{p^2} z^{p^3}+\mu^p(a^pz^{p^2}+a^{p^3}z^p)+\mu a^{p^2}z=0\\
&\Longleftrightarrow \mu^{p^2}a^{p^2} z^{p^3}+\mu^pa^pz^{p^2}+\mu^pa^{p^3}z^p+\mu a^{p^2}z=0\\
&\Longleftrightarrow (\mu^{p^2}a^{p^2} z^{p^3}+\mu^pa^{p^3}z^p)+(\mu^pa^pz^{p^2}+\mu a^{p^2}z)=0\\
&\Longleftrightarrow (\mu^pa^pz^{p^2}+\mu
a^{p^2}z)^p+(\mu^pa^pz^{p^2}+\mu a^{p^2}z)=0,
\end{align*}
i.e.
\begin{equation}\label{Lg}
(\mu^pa^pz^{p^2}+\mu a^{p^2}z)\in \mathbb{F}_p=\GF{r}.
\end{equation}

Given $\gcd(n,r)=1$, since $\GF{n}\cap \GF{r}=\{0,1\}$, \eqref{Lg}
means that $\mu^pa^pz^{p^2}+\mu a^{p^2}z=0$ or $\mu^pa^pz^{p^2}+\mu
a^{p^2}z=1$.  When $z\neq0$, we have
\[
\mu^pa^pz^{p^2}+\mu a^{p^2}z=0 \Longleftrightarrow
z^{p^2-1}=\left(\frac{a^p}{\mu}\right)^{p-1} \Longleftrightarrow
z^{p+1}=\frac{a^p}{\mu},
\]
where it was regarded $(p-1,2^n-1)=1$ which follows from
$\gcd(n,r)=1$.

Consequently,  $\varepsilon_{g_\mu,a}=K_{a,1}\cup K_{a,2}\cup
K_{a,3}$, where $K_{a,1}=\{x\in\GF{n}|ax^p+a^px=0\}$,
$K_{a,2}=\{x\in\GF{n}|z^{p+1}=\frac{a^p}{\mu}, z=ax^p+a^px\}$,
$K_{a,3}=\{x\in\GF{n}|z^{p^2}+\left(\frac{a^p}{\mu}\right)^{p-1}z+\frac{1}{\mu^pa^p}=0,
z=ax^p+a^px \}$.

Now, we need following fact.
\begin{lemma}\label{gcd_p+1_general}(Lemma 11.1 in \cite{23}) For $1\leq r \leq n$,
\[
\gcd(2^r+1, 2^n-1)=
\begin{cases}
1, &\text{ if }\gcd(2r,n)=\gcd(r,n)\\
2^{\gcd(r,n)}+1, &\text{ if }\gcd(2r,n)=2\gcd(r,n).
\end{cases}
\]
\end{lemma}
Therefore, when $\gcd(n,r)=1$,
\begin{equation}\label{gcd_p+1_2^n-1}
\gcd(p+1, 2^n-1)=
\begin{cases}
1, &\text{ if $n$ is an odd}\\
3, &\text{ if $n$ is an even}.
\end{cases}
\end{equation}
Since
$K_{a,1}=\{x\in\GF{n}|(\frac{x}{a})^p+\frac{x}{a}=0\}=\{x\in\GF{n}|\frac{x}{a}\in\GF{r}\}=\{x\in\GF{n}\cap
a\GF{r}\}=\{0,a\}$, for every $z\in\GF{n}$, the linear equation
$z=ax^p+a^px$ has at most two solutions. By using Lemma
\ref{lemrootnum}, we can see:
\begin{equation}\label{k2k3_g}
|K_{a,2}|\leq 6, |K_{a,3}|\leq 8
\end{equation}
 and that Lemma
\ref{parity} let us know that
\begin{equation}\label{k2+k3_g}
|K_{a,2}|+|K_{a,3}|=
\begin{cases}
2 \text{ or } 14, &\text{ if $n$ is even;}\\
0 \text{ or } 6, &\text{ if $n$ is odd.}
\end{cases}
\end{equation}

On the other hand, when $\gcd(n,r)=1$, if the equation $z=ax^p+a^px$
for $z\in \GF{n}$ has a solution $x\in \GF{n}$, then
$Tr(\frac{z}{a^{p+1}})=0$. The reverse of this proposition is no
generally validate and thus it seems hard to get the exact
distribution of $|K_{a,2}|$ as done in Subsection \ref{special sec}.

However the exactly same lower-bound-estimations as in Subsection
\ref{special sec} still hold as described below.
 To begin with, let us note $\gcd(p^2+p+1,
2^n-1)=\gcd((p^2+p+1)(p-1), 2^n-1)=\gcd(p^3-1,
2^n-1)=2^{\gcd(3,n)}-1$.

\begin{enumerate}
\item For even $n$, there are
$\frac{2(2^n-1)}{3}$ $a$'s such that $\frac{a^p}{\mu}$ is not a
$(p+1)$-th power (or, by \eqref{gcd_p+1_2^n-1}, equivalently,
$\frac{a^p}{\mu}$ is a non-cubic) in $\GF{n}$, i.e, $|K_{a,2}|=0$
(in this case $|K_{a,3}|=2$ by \eqref{k2k3_g} and \eqref{k2+k3_g},
and $r_a=2$). That is, there are at most $\frac{(2^n-1)}{3}$ $a$'s
such that $r_a=4$.

Furthermore, if $3\dag n$ and therefore the $(p^2+p+1)$-th powering
is a permutation of $\GF{n}$, then for any  $a\neq a'\in \GF{n}^*$
such that $\frac{a^p}{\mu}$ and $\frac{a'^p}{\mu}$ are $(p+1)$-th
powerings, when $b^{p+1}=\frac{a^p}{\mu}$ and
$b'^{p+1}=\frac{a'^p}{\mu}$, one has $\frac{b}{a^{p+1}}\neq
\frac{b'}{a'^{p+1}}$, because $(p+1)-$th powering to the both side
of $\frac{b}{a^{p+1}}= \frac{b'}{a'^{p+1}}$ leads to
$a^{p^2+p+1}=a'^{p^2+p+1}$ i.e. $a=a'$ i.e. a contradiction. Thus,
when $a$ takes all elements of $\GF{n}^*$ such that
$\frac{a^p}{\mu}$ are $(p+1)$-th powerings and $b$ takes all three
$(p+1)-$th power roots of $\frac{a^p}{\mu}$, $\frac{b}{a^{p+1}}$
takes all $2^n-1$ elements in $\GF{n}^*$. On the other hand, by
\eqref{k2k3_g} and \eqref{k2+k3_g}, if $r_a=4$, then $K_{a,2}=6$ and
so it must be true that $Tr(\frac{b}{a^{p+1}})=0$ for all three
$(p+1)-$th power root $b$'s of $\frac{a^p}{\mu}$. Since in
$\GF{n}^*$ there are $2^{n-1}-1$ elements with absolute trace 0, it
follows that there are only at most $\frac{(2^{n-1}-2)}{3}$ $a$'s
with $r_a=4$.
\item For odd $n$, by
\eqref{gcd_p+1_2^n-1} every element of $\GF{n}$ is a $(p+1)$-th
power and it holds
\begin{equation}\label{trace}
Tr(\frac{z}{a^{p+1}})=Tr\left(\frac{(\frac{a^p}{\mu})^{1/(p+1)}}{a^{p+1}}\right)=Tr\left(\left(\frac{1}{\mu
a^{p^2+p+1}}\right)^{\frac{1}{p+1}}\right).
\end{equation}
First, we will show that $|K_{a,3}|=6$ can not happen. Let us
suppose the opposite: $|K_{a,3}|=6$. This is possible only when the
equation
$z^{p^2}+\left(\frac{a^p}{\mu}\right)^{p-1}z+\frac{1}{\mu^pa^p}=0$
has 4 solutions $z_1,z_2,z_3,z_4$ (please, regard Lemma
\ref{lemrootnum}) and for exactly one (assuming it is $z_4$ wlog )
among these solutions the equation $z_4=ax^p+a^px$ has no solution,
which is a contradiction because given $x_1,x_2,x_3$ that are
solutions of $z_1=ax^p+a^px,z_2=ax^p+a^px, z_3=ax^p+a^px$
respectively, $x=x_1+x_2+x_3$ is a solution of $z_4=ax^p+a^px$
(since $z_4=z_1+z_2+z_3$).

\begin{enumerate}
\item If $n\equiv \pm1 \mod 6$, then $\gcd(p^2+p+1, 2^n-1)=1$ and by
\eqref{trace} there are exactly $2^{n-1}$ $a$'s such that
$Tr(\frac{z}{a^{p+1}})=1$ for
$z=\left(\frac{a^p}{\mu}\right)^{\frac{1}{p+1}}$. Thus, there are at
least $2^{n-1}$ $a$'s such that $|K_{a,2}|=0$ (in this case, by
\eqref{k2+k3_g} $|K_{a,3}|=0$ and so $r_a=1$).

\item If $n\equiv 3 \mod 6$, then $\gcd(p^2+p+1, 2^n-1)=2^3-1=7$ and
therefore there are exactly $wt(f_7)$ $a$'s such that
$Tr(\frac{z}{a^{p+1}})=1$ for
$z=\left(\frac{a^p}{\mu}\right)^{\frac{1}{p+1}}$. Thus, there are at
least $wt(f_7)$ $a$'s such that $|K_{a,2}|=|K_{a,3}|=0$ and $r_a=1$.
\end{enumerate}
The exactly same derivation as done in Subsection \ref{special sec}
gives:
\end{enumerate}

\begin{theorem} Let $g_{\mu}=Tr(\mu x^{2^{2r}+2^{r}+1})$, where $\mu \in \GF{n}$
$\gcd (n,r)=1$ and $n\geq 4$.
\begin{enumerate}
\item  For $n\equiv 2,4 \mod 6$,
\begin{equation}\label{gmu-n24}
nl_2(g_\mu)\geq
2^{n-1}-\frac{1}{2}\sqrt{2^n+\frac{7}{3}\cdot2^{\frac{3n}{2}}-\frac{10}{3}2^{\frac{n}{2}}}.
\end{equation}
\item  For $n\equiv 0 \mod 6$,
\begin{equation}\label{gmu-n0}
nl_2(g_\mu)\geq
2^{n-1}-\frac{1}{2}\sqrt{2^n+\frac{8}{3}\cdot2^{\frac{3n}{2}}-\frac{8}{3}2^{\frac{n}{2}}}.
\end{equation}
\item If $n\geq 5$ and $n\equiv \pm 1 \mod 6$, then
\begin{equation}\label{gmu-n15}
nl_2(g_\mu)\geq
2^{n-1}-\frac{1}{2}\sqrt{2^n+3\cdot2^{\frac{3n+1}{2}}-2^{\frac{n+3}{2}}}.
\end{equation}
\item If $n\equiv 3 \mod 6$ and $n\geq 5$, then
\begin{equation}\label{gmu-n3}
nl_2(g_\mu)\geq
2^{n-1}-\frac{1}{2}\sqrt{2^n+(2^n-1)2^{\frac{n+3}{2}}-wt(f_7)2^{\frac{n+1}{2}}}.
\end{equation}
\end{enumerate}
\end{theorem}

As evident, the new obtained lower bounds are significantly better
than ones given by Theorem \ref{Gode-Gango-G}.

\subsection{Second-Order Nonlinearity of $g_\mu$ with $\gcd(n,r)\neq 1$}
If $n=3r$, then \eqref{maineq} reduces to $
(\mu^{p^2}+\mu^p+\mu)(ax^p+a^px)=0$ and therefore has $p$ solutions
(to be precise, under the condition $Tr_{r}^{n}(\mu)\neq 0$), that
is, $r_a\leq r$ for every $a\in \GF{n}$. So, the lower bound stated
in the item 1 of Example \ref{manypaper} follows.

 When $\gcd(n,r)\neq 1$ and $n\neq 3r$, from
\eqref{Lg} it follows that $\varepsilon_{g_\mu,a}$ is the solution
set of
\begin{align*}
z=ax^p+a^px, \prod_{\omega\in \GF{\gcd(n,r)}}(z^{p^2}+
\left(\frac{a^p}{\mu}\right)^{p-1}z+\frac{\omega}{\mu^pa^p})=0.
\end{align*}

Consequently,  $\varepsilon_{g_\mu,a}=K_{a,1}\cup K_{a,2}\cup
K_{a,3}$, where $K_{a,1}=\{x\in\GF{n}|ax^p+a^px=0\}$,
$K_{a,2}=\{x\in\GF{n}|z^{p+1}=\frac{a^p}{\mu}\GF{\gcd(n,r)}^*,
z=ax^p+a^px\}$, $K_{a,3}=\{x\in\GF{n}|\prod_{\omega\in
\GF{\gcd(n,r)}^*}(z^{p^2}+
\left(\frac{a^p}{\mu}\right)^{p-1}z+\frac{\omega}{\mu^pa^p})=0,
z=ax^p+a^px \}$.

Since
$K_{a,1}=\{x\in\GF{n}|(\frac{x}{a})^p+\frac{x}{a}=0\}=a\GF{\gcd(n,r)}$,
for every $z\in\GF{n}$, the linear equation $z=ax^p+a^px$ has at
most $2^{\gcd(n,r)}$ solutions. And, if the linear equation
$z^{p^2}+
\left(\frac{a^p}{\mu}\right)^{p-1}z+\frac{\omega}{\mu^pa^p}=0$ has a
solution in $\GF{n}$, then it has the same number of solutions as
$z^{p^2}+ \left(\frac{a^p}{\mu}\right)^{p-1}z=0$ has  in $\GF{n}$,
i.e. $z^{p+1}=\frac{a^p}{\mu}\GF{\gcd(n,r)}^*$ or $z=0$.

Corollary 1 and Corollary 2 of \cite{19} states the upper bound on
root number of the special linearized polynomial $z^{p^2}+az^p+bz$
where $a,b\in\GF{n}$, $p=2^r$ and $\gcd(n,r)=1$, to be 4. When
$a=0$, but without the restriction $\gcd(n,r)=1$, we can get the
exact root number by using Lemma \ref{gcd_p+1_general}.
\begin{proposition}
For the linearized polynomial $z^{p^2}+bz$ where $b\in\GF{n}^*$ and
$p=2^r$, its root number is
\begin{enumerate}
\item 1 if  $b$ is not a $(p^2-1)-$power in $\GF{n}$; \\
\item $2^{\gcd(n,r)}$ if $\|n\|_2\geq \|r\|_2$ and  $b$ is a $(p-1)-$power (so also a $(p^2-1)-$power) in $\GF{n}$; \\
\item
$2^{2\gcd(n,r)}$ if $\|n\|_2< \|r\|_2$  and $b$ is a $(p^2-1)-$power in $\GF{n}$.\\
\end{enumerate}
\end{proposition}

From the facts mentioned above, following inequalities follow.
\begin{equation*}
|K_{a,2}|\leq \begin{cases}
2^{\gcd(n,r)}(2^{\gcd(n,r)}-1), &\text{ if }\|n\|_2\geq\|r\|_2\\
2^{\gcd(n,r)}(2^{2\gcd(n,r)}-1), &\text{ if }\|n\|_2<\|r\|_2.
\end{cases}
\end{equation*}

\begin{equation*}
|K_{a,3}|\leq \begin{cases}
2^{2\gcd(n,r)}(2^{\gcd(n,r)}-1), &\text{ if }\|n\|_2\geq\|r\|_2\\
2^{3\gcd(n,r)}(2^{\gcd(n,r)}-1), &\text{ if }\|n\|_2<\|r\|_2.
\end{cases}
\end{equation*}

Thus
\[
 r_a\leq  \begin{cases}
 3\gcd(n,r), &\text{ if }\|n\|_2\geq\|r\|_2\\
 4\gcd(n,r), &\text{ if }\|n\|_2<\|r\|_2,
 \end{cases}
 \] and by using Lemma \ref{parity} we improve on the lower bound \eqref{i+j} as follows:
\begin{equation}\label{i+j_1}
nl_2(g_\mu)\geq
\begin{cases}
2^{n-1}-\frac{1}{2}\sqrt{2^n+(2^n-1)2^{\lfloor\frac{n+3gcd(n,r)}{2}\rfloor}}, &\text{ if }\|n\|_2\geq\|r\|_2\\
2^{n-1}-\frac{1}{2}\sqrt{2^n+(2^n-1)2^{\lfloor\frac{n+4gcd(n,r)}{2}\rfloor}},
&\text{ if }\|n\|_2<\|r\|_2.
 \end{cases}
\end{equation}

Since
$\gcd(p+1,2^{\gcd(n,r)}-1)|\gcd(p+1,2^r-1)=\gcd(2^r+1,2^r-1)=1$,
every element of $\GF{\gcd(n,r)}^*$ has unique $(p+1)$-th power root
in the field itself. Hence, when $\|n\|_2<\|r\|_2$, for the
$\frac{2^{\gcd(n,r)}}{2^{\gcd(n,r)}+1}(2^n-1)$ $a$'s such that
$\frac{a^p}{\mu}$ is not a $(p+1)$-th power of some entry in
$\GF{n}$), the equation $z^{p+1}=\frac{a^p}{\mu}\GF{\gcd(n,r)}^*$
has no solution, and so  $z^{p^2}+
\left(\frac{a^p}{\mu}\right)^{p-1}z+\frac{\omega}{\mu^pa^p}=0$ for
any $\omega\in \GF{\gcd(n,r)}^*$ has at most one solution. Thus,
when $\|n\|_2<\|r\|_2$, for such
$\frac{2^{\gcd(n,r)}}{2^{\gcd(n,r)}+1}(2^n-1)$ $a$'s,

\[
|K_{a,2}|=0, |K_{a,3}|=2^{\gcd(n,r)}(2^{\gcd(n,r)}-1)
\]
and
\[
r_a\leq 2\gcd(n,r).
\]

By \eqref{Mesnager}, when $\|n\|_2<\|r\|_2$ (note that in  this case
$n$ is even ), we get
\begin{align*}
nl_2(g_\mu)&\geq
2^{n-1}-\frac{1}{2}\sqrt{2^n+(2^n-1)(\frac{1}{2^{\gcd(n,r)}+1}2^{\frac{n+4\gcd(n,r)}{2}}+\frac{2^{\gcd(n,r)}}{2^{\gcd(n,r)}+1}2^{\frac{n+2\gcd(n,r)}{2}}})\\
&=2^{n-1}-\frac{1}{2}\sqrt{2^n+\frac{2^{2\gcd(n,r)+1}}{{2^{\gcd(n,r)}+1}}(2^{\frac{3n}{2}}-2^{\frac{n}{2}})}.
\end{align*}

\begin{theorem}\label{gcdneq1} For $g_{\mu}=Tr(\mu x^{2^{2r}+2^{r}+1})$, where $\mu \in
\GF{n}$,
$\gcd (n,r)\neq 1$ and $n\geq 4$.
\begin{equation}
nl_2(g_\mu)\geq
\begin{cases}
2^{n-1}-\frac{1}{2}\sqrt{2^n+(2^n-1)2^{\lfloor\frac{n+3gcd(n,r)}{2}\rfloor}}, &\text{ if }\|n\|_2\geq\|r\|_2\\
2^{n-1}-\frac{1}{2}\sqrt{2^n+(2^n-1)2^{\frac{n}{2}}\frac{2^{2\gcd(n,r)+1}}{{2^{\gcd(n,r)}+1}}},
&\text{ if }\|n\|_2<\|r\|_2.
 \end{cases}
\end{equation}
\end{theorem}
This lower bound is better than one which we showed in \eqref{i+j}
in particular as $gg(n,r)\geq \gcd(n,r)$.
\begin{corollary}
If $n=sr$ where $s$ is an odd greater than 3, then
\begin{equation}
nl_2(g_\mu)\geq
2^{n-1}-\frac{1}{2}\sqrt{2^n+(2^n-1)2^{\frac{n+3r}{2}}}
\end{equation}

If $n=sr$ where $s$ is an even greater than 2, then
\begin{equation}
nl_2(g_\mu)\geq
2^{n-1}-\frac{1}{2}\sqrt{2^n+(2^n-1)2^{\frac{n}{2}}\frac{2^{2r+1}}{{2^{r}+1}}}
\end{equation}
\end{corollary}

The lower bounds presented by this corollary are better than ones
given by Items 2-4 of Example \ref{manypaper} which can be
reformulated as: For $n=sr, 4\leq s\leq 6$
\[
nl_2(g_\mu)\geq
2^{n-1}-\frac{1}{2}\sqrt{2^n+(2^n-1)2^{\frac{n}{2}}2^{(\frac{s}{2}-1)r}}.
\]

On the other hand, when $s=7$, this corollary gives the same lower
bound with Corollary \ref{manypaper_new}.

\subsection{Problems for further considerations}

If $Tr(\frac{z}{a^{p+1}})=1$, then the  equation $z=ax^p+a^px$ has
no solution in $\GF{n}$.
\begin{problem}
Use this fact to improve on the lower bound of second-order
nonlinearity given in Theorem \ref{gcdneq1} for $g_{\mu}=Tr(\mu
x^{2^{2r}+2^{r}+1})$, where $\mu \in \GF{n}$, $\gcd (n,r)\neq 1$ and
$n\geq 4$.
\end{problem}
 Consider generic cubic monomial Boolean function
$f_{\mu}=Tr(\mu x^{2^{i}+2^{j}+1})$, where $\mu \in \GF{n}$ and
$n>i>j>0$. Let us introduce denotations: $p=2^j, q=2^i$. The
quadratic part of derivative $D_af_\mu$ of $f_\mu$ at $a\in \GF{n}$
is $Tr(\mu ax^{q+p}+\mu a^px^{q+1}+\mu a^{q}x^{p+1})=Tr(\mu^{1/p}
a^{1/p}x^{q/p+1}+\mu a^px^{q+1}+\mu a^{q}x^{p+1})$. With reference
to \eqref{polarform},  $\varepsilon_{f_\mu,a}$ is the solution set
of linear equation
\[
\mu^{1/p} a^{1/p}x^{q/p}+\mu a^px^{q}+\mu a^{q}x^{p}+\mu^{1/q}
a^{1/q}x^{p/q}+\mu^{1/q} a^{p/q}x^{1/q}+\mu^{1/p}a^{q/p}x^{1/p}=0,
\]
or, equivalently
\begin{equation}\label{openP}
L(x)=[(a\mu)x^p+(a^p\mu)x+(a^{q^2}\mu^q)x^{pq}]^p+[(a\mu)x^q+(a^q\mu)x+(a^{p^2}\mu^p)x^{pq}]^q=0.
\end{equation}

\begin{problem}
Determine the set of $a$'s such that the equation \eqref{openP} has
solutions of smaller number than $2^V$ in $\GF{n}$ where $V$ is
given by Theorem \ref{corrootnum} (or computed by Section
\ref{secmini}).
\end{problem}

\section{Conclusion}
When a linearized polynomial is given, to determine its root number
is an important task in finite field and symmetric cryptography
theory. This paper contributes to give a better general method to
get more precise upper bound on the root number of any given
linearized polynomial.

Then, as an application of this result, we improve the estimation
for lower bound of the second-order nonlinearities of cubic Boolean
functions. For example, for cubic monomial Boolean function
$f_\mu(x)=Tr(\mu x^{2^9+2^5+1})$, the best previous result \cite{8}
can say $nl_2(f_\mu)\geq 393216$ over $F_{2^{20}}$ and
$nl_2(f_\mu)\geq 196608$ over $F_{2^{19}}$. By this paper, now we
know $nl_2(f_\mu)\geq 431605$ over $F_{2^{20}}$ and $nl_2(f_\mu)\geq
238971$ over $F_{2^{19}}$. And, while the best previous result can
show only $nl_2(Tr(\mu x^{2^{18}+2^{10}+1}))\geq 76781$ over
$F_{2^{19}}$, this paper proves $nl_2(Tr(\mu
x^{2^{18}+2^{10}+1}))\geq 238971$.

Furthermore, this paper shows that  by studying the distribution of
radicals of derivatives of a given Boolean functions one can get a
better lower bound of the second-order nonlinearity, through an
example of the Boolean function $g_{\mu}=Tr(\mu x^{2^{2r}+2^r+1})$
over any finite field $\GF{n}$.

These results show that many cubic Boolean functions such as
$g_{\mu}=Tr(\mu x^{2^{2r}+2^r+1})$ over any finite field $\GF{n}$
have larger Hamming distance to the affine functions and quadratic
functions than it was known (thus could be expected). They can be
used in choice of cubic Boolean functions which are resistant
against linear and quadratic approximation attacks.

\end{document}